\documentclass[journal]{IEEEtran}
\usepackage{booktabs} 
\usepackage{tabularx}
\usepackage{color}
\usepackage[strict]{changepage}
\usepackage{calc}
\usepackage{mathtools,lipsum}
\usepackage{amsmath}
\usepackage{amssymb}
\usepackage{mathrsfs}
\usepackage{upgreek}
\usepackage{graphicx}
\usepackage{algorithmicx}
\usepackage{setspace}
\usepackage{algorithm}
\usepackage{enumitem}
\usepackage[noend]{algpseudocode}
\usepackage{caption}
\usepackage{subcaption}
\usepackage{amsthm, times, amsfonts, cite}

\newtheorem{theorem}{Theorem}
\newtheorem{lemma}{Lemma}

\usepackage[T1]{fontenc}
\usepackage[latin9]{inputenc}
\usepackage{array}
\usepackage{float}
\usepackage{amsmath}
\usepackage{amssymb}
\usepackage{graphicx}

\usepackage{amsmath}

\begin{document}

\title{
 Joint Information Freshness and Completion Time Optimization 
  for Vehicular Networks
%
%
}

%
\author{Abubakr Alabbasi and Vaneet Aggarwal\\
Purdue University, West Lafayette, IN 47907, USA\\
 Email:\{aalabbas,vaneet\}@purdue.edu \\
}
	\maketitle

\begin{abstract}

The demand for real-time cloud applications has seen an unprecedented growth over the past  decade. These applications require rapidly data transfer and fast computations. This paper considers a scenario where multiple IoT devices update information on the cloud, and request a computation from the cloud at certain times. The time required to complete the request for computation includes the time to wait for computation to start on busy virtual machines, performing the computation, waiting and service in the networking stage for delivering the output to the end user. In this context, the freshness of the information is an important concern and is different from the completion time. This paper proposes novel scheduling strategies for both computation and networking stages. Based on these strategies, the age-of-information (AoI) metric and the completion time are characterized. A convex combination of the two metrics is optimized over the scheduling parameters. The problem is shown to be convex and thus can be solved optimally. Moreover, based on the offline policy, an online algorithm for job scheduling is developed. Numerical results demonstrate significant improvement as compared to the considered baselines.
\end{abstract}


\section{Introduction}


Traditional cloud services have been mainly designed to achieve high throughput and small delay \cite{bedewy2016optimizing,sarathy2010next}. However, these performance measures fail to capture the timeliness of the information from the application perspective which is important for real-time cloud, and IoT (Internet-of-things), applications. For instance, a message with stale information is of a little value, even if it is promptly delivered to the destination. In contrast, fresh information would have much greater merit, even if it is somewhat  delayed. As an example, to perform rate adaption or resource scheduling in cellular networks, a timely channel state information is important especially in rapid varying channels, which is a typical case in new 5G technologies such as Millimeter wave systems \cite{joint013}. Further, when navigating directions in an autonomous vehicle, timely information of other vehicles is important for the computation. This paper aims to provide scheduling strategies for computation and networking so as to obtain timely information in addition to prompt delivery of information. 

 The freshness of delivered information is measured by a new metric, so-called, age-of-information (AoI) \cite{kaul2011minimizing,sun2017update, huang2015optimizing, chen2016age, kam2013age}. AoI metric measures the freshness of information updates, which is typically defined as the time  elapsed since the last delivered information update was generated (at the source). The AoI of a particular job increases linearly as time goes by until the job has been fully processed by the server and delivered to the target destination \cite{kaul2012real}. Upon reception, the age drops to the time elapsed since the generation of the information. We note that the AoI is different from completion time because AoI  depends on the freshness of information, and not the overall time taken from request to the delivery.

 Unlike packet-centric measures like throughput or latency, AoI is a destination-centric metric which makes it more appropriate to characterize the freshness and timeliness of information updates. We note that age is fundamentally different from traditional performance measures as of throughput and latency. To reduce the latency, information updates should not be sent frequently to reduce the system overload and to avoid system's congestion (i.e., resulting a low latency). Nonetheless, the information at the destination could be stale due to a lack of fresh updates. Intuitively, the higher the rate updates, the fresh the information is (i.e., age will be small since more fresh updates are executed and delivered).
 In contrast, the increased update rates (i.e., a larger throughput) would make the load closer to the server's capacity, and thus can increase the backlog of the queue. As a result, requests may have larger delay and hence, the information again becomes stale when arrived at destination. Thus, minimizing age is fundamentally different from maximizing throughput or minimizing latency.

Recently, AoI metric has been studied for multiple applications, including  vehicular networks \cite{kaul2011minimizing,kaul2012real}, feedback for wireless channels \cite{costa2015age,talak2018optimizing}, cloud gaming \cite{yates2017timely}, and mobile caching \cite{kam2017information}. 
 In \cite{kaul2011minimizing}, for example, AoI is investigated for vehicular networks where a single link, modeled as a single queue, is considered. It is shown that age can be minimized by controlling the queueing discipline, e.g., using last-in-first-out (LIFO) scheme rather than simple First-In-First-Out (FIFO) scheme. However, this analysis was limited to a single link. 
In \cite{kaul2012real}, similar analysis for characterizing AoI for a single link was performed under different queueing disciplines including M/M/1, M/D/1, and D/M/1. This analysis is generalized to multicast in \cite{huang2015optimizing}, considering M/G/1 and G/G/1 queues. The impact of packet drops on age is studied in \cite{chen2016age}, while the effect of out of order deliveray on AoI is provided in \cite{kam2013age}. Further, age was investigated for different arrival and service time distribution in several studies, see for example \cite{kaul2012status} and references therein. 
However, there is no work that considers computing of the jobs on the data center and the delivery of  the output to the end user, to the best of our knowledge. This is the focus of the current paper where the trade-off between completion time and age of information is investigated.


Consider a scenario where different sensors (IoT devices) are updating the content at the server. Each sensor also requests a computation at a certain rate, which includes the information from multiple sensors. Since the data center consists of multiple virtual machines (VMs), a VM must be selected for each job to perform computation. Each computation VM has a queue for processing jobs and after the computation is performed, the output must be transmitted to the user. In order to send the computed result, the result waits in a networking queue which is followed by transmission. This can be applied for navigation of autonomous vehicles where the job of one vehicle depends on the update of other vehicles. The computation and the networking are time consuming. The results are needed with a low completion time. In addition to that, the freshness of the result is also important. We weigh the completion and the freshness of information of different jobs in order to design the scheduling strategies for both the computation and the networking phases. 

We note that the time in the networking phase impacts both the AoI as well as the completion time. The objective of the networking phase scheduling is to minimize the weighted completion time of the jobs. If the arrival rate in the networking phase for each type of jobs is Poisson, priority scheduling is optimal \cite{pinedo2016scheduling}, where the priority is based on the weighted shortest expected processing time (WSEPT) of the jobs. However,  the networking phase uses the  output of the computation phase scheduling and the arrival may not be Poisson in general for the networking phase even if different types of jobs arrive to the system as Poisson process. Thus, the scheduling for the computation phase is important. Since the computation phase has multiple VMs, even the weighted completion time is an NP-hard problem \cite{pinedo2016scheduling}. Further, the AoI metric does not depend on the waiting time in the queue of the VM making the problem challenging. Since the problem of finding optimal assignment is challenging, we use a probabilistic scheduling approach where the VMs are assigned to jobs with certain probabilities which can be optimized for improved performance. Such scheduling approaches have been used in \cite{xiang2016joint}. This approach results in the arrival distribution for the different types of jobs in the networking phase also being a Poisson process.

We aim to fundamentally understand the  tradeoff between completion time and freshness of information. We  formulate the problem as a convex combination of both average completion time and mean AoI metrics. The key parameter is the probabilities of scheduling the VMs in the computation phase.   Thus, we optimize this convex combination over the choice of probabilistic scheduling parameters. The optimization problem is shown to be convex and thus can be efficiently solved using any convex optimization solver. 
Numerical results demonstrate significant improvement
of AoI metric as compared to the considered baselines. In particular, our proposed approach shows an improvement of 25\% better than the most competitive baseline.
 The
key contributions of our paper are summarized as follows.    

\begin{itemize}[leftmargin=0.3cm,itemindent=.3cm,labelwidth=\itemindent,labelsep=0cm,align=left]
	\item We consider a scenario where multiple sensors update the information in the data center. Further, each sensor requests for computation at a certain rate where the information from other sensors is needed. Both the computation and transmission of the output are considered together. This is the first paper, to the best of our knowledge, that considers AoI as well as completion time of the jobs in such IoT  applications. 
	

	\item A novel  probabilistic scheduling policy is proposed  to assign each job to VMs so that the overall objective is achieved. Based on the scheduling approaches, the mean AoI and the average completion time are characterized.

	\item  A holistic optimization framework is developed to optimize a convex  combination of a weighted sum of completion time and age. The problem is shown to be convex and thus can be efficiently solved using Projected Gradient Descent Algorithm.  Further, based on the offline algorithm, an online version is developed to keep tack of the system dynamics and thus improves the system performance. The online algorithm is evaluated on  MSR Cambridge public Traces  \cite{narayanan2008write}, and show that online algorithm achieves comparable performance to the offline algorithm (which uses Poisson arrival distribution with known parameters). 
	\item Numerical results demonstrate a significant improvement of AoI metric as compared to the considered baselines. Further, an efficient tradeoff point between the two metrics can be chosen such that a reasonable level of staleness can be tolerated by the application.
	
\end{itemize}

The rest of this paper is organized as follows. We provide the system description in Section \ref{sysModel}, 
where we also characterize the age metric, probabilistic scheduling and queueing model.
Age analysis is presented in Section \ref{AoI_analysis}. Further, in this section, we provide  expressions for completion time and expected age of information, for any job, at the user's side. In Section \ref{proForm}, we formulate the optimization problem and present the proposed algorithm to optimally solve the joint convex combination of age and completion time metrics. 
Numerical results are presented in Section \ref{simRes}, and we conclude in
Section \ref{conc}.

\section{System Model}
\label{sysModel}

\subsection{Target Systems: Information-update-based Systems}

Our system is motivated by the unprecedented growth in the demand for real-time IoT based cloud computing applications. These applications are carried over distributed stream processing computation frameworks, e.g., Apache Storm/Spark \cite{iqbal2015big}. Many emerging applications, such as vehicular networks \cite{kaul2011minimizing} and cloud gaming \cite{yates2017timely}, do not require only rapidly computing data streams but also need fresh and accurate information. For these applications, not only the performance (e.g., throughput and completion time) of the computing systems is of a major concern but also the freshness (e.g., age) of the delivered results is of a primary concern \cite{sun2017update}. Clearly, optimizing each one of these dimensions (completion time and age) separately would lead to different set of outcomes, possibly conflicting decisions. To further illustrate, let us consider age of information and completion time of a certain job. A service provider needs to complete jobs as soon as possible to avoid congestion. On the other hand, the applications need a fresh information to make right decisions. Therefore, joint optimization for both metrics is very important to better understand the tradeoff among different dimensions and, then, come up with simple yet efficient algorithms that optimize such tradeoff for real time cloud applications.
The focus of this paper is to investigate such tradeoff and to develop solutions that will address the different angles in such design problems. 
 Next, we explain the update process and some related assumptions.


\subsection{Information Updates}


We assume that there are $J$ users/vehicles in the system. Users want to get the most fresh versions of the route-related data to be aware of any sudden change that would occur on their ways to destinations. In vehicular networks, exchanging positions, velocity, direction to destination and control information in real time is critical to safety and collision avoidance \cite{kaul2011minimizing}.  
To achieve this goal, each user (vehicle) has not only to know its local information but also to know the status (e.g., position, velocity, etc) of  a set (or all) of other vehicles that help determine the route information. Each vehicle updates the data on the cloud at certain rate, and requests computation for the tasks at certain rate.  Since the number of vehicles is typically large and the tasks need to be executed as soon as they are requested, we assume that 
every device/vehicle sends its recent status data (e.g., location, speed, direction,  etc.) asynchronously to the server to update its stored information. This information is updated frequently and the server has to maintain the last updated information to perform certain tasks whenever required/needed. Such real-time information updates are also found in many similar applications including control systems, autonomous cars, sensor networks, on-line gaming,  hazard signals and any timely-status update systems.

We assume that the set of vehicles is $\mathbb{J}= \{1, \cdots, J\}$. Let $j$ denote a vehicle in the set $\mathbb{J}$, i.e., $j\in \mathbb{J}$.  We further assume that the information updates, for every vehicle, follow a Poisson process with rate $\mu_j >0$.
This information updating process is independent across the different vehicles $j$'s and, thus, the inter-update time, for every vehicle $j$, at the server is exponentially distributed with rate $\mu_j$. Note that inter-update time is the time required for the server to receive the updated-information of the vehicles. In such a shared environment, the server has to use not only requested vehicle's information  to take an action (e.g., change the direction or the speed) but needs to know all nearby vehicles information and the data of the vehicles in its route to determine the future decisions and/or consequently its path to the destination. Thus, besides having a high speed server, a centralized system with global information of a subset (or all) of vehicles in the environment is needed to perform the computing tasks. In order to do such computationally expensive tasks, a central unit has to first collect all necessary information from the agents (vehicles) and then performs the needed tasks on the fly. Further, global information of all related vehicles has to be available at the central unit (e.g., server farm) so that the decision is accurately taken.  For example, to change the direction or the route of a particular vehicle, in case of a congestion, collision, or even for a temporary blocked route, the exact locations and the status of other vehicles in the way to destination have to be efficiently exploited to figure out what exact process needs to be calculated, e.g.,  the required  time to arrive, new velocity, direction, etc. Thus, the server has to use the most recent/fresh information to properly generate certain decisions. In the next subsections, we explain our models for computing (i.e., performing the computational part of the the job) and networking (routing the data back to the corresponding vehicle).

\subsection{System Parameters and Computing Model}

\textcolor{black}{
\begin{figure}[t]
	\centering\includegraphics[trim=0.5in 0.15in 2.05in 0.1in, clip,width=0.491\textwidth]{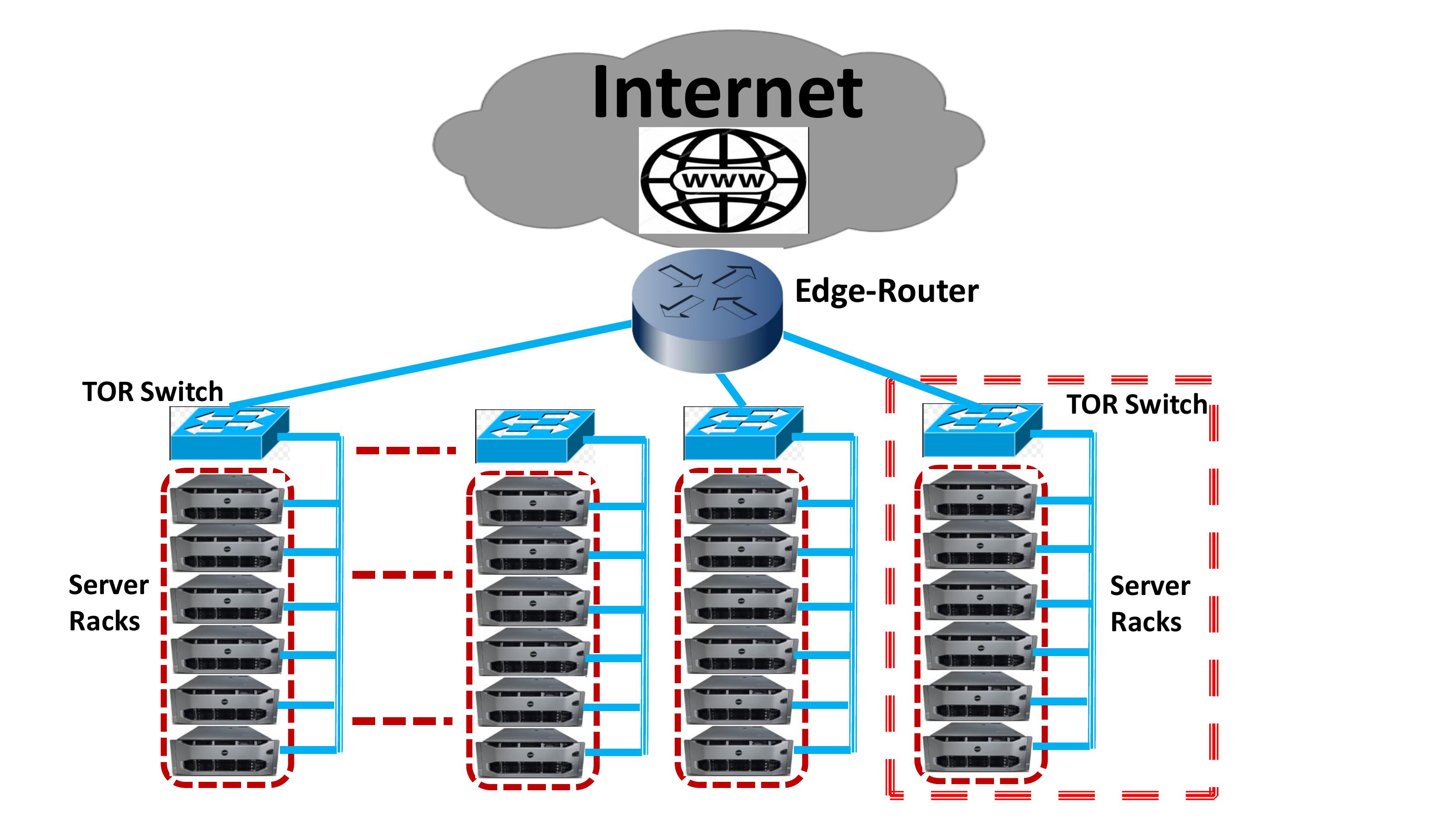}
	\caption{\textcolor{black}{A schematic illustrates a server farm, composed of several server racks, top-of-rack switches (TORs), and an edge-router. While server farms can have more layers of TOR switches and/or edge-routers, we abstract our infrastructure in a per-site basis, without loss of generality,  to only one layer of TORs.
		\label{sysModel1}}}
\end{figure}
}

\textcolor{black}{
	\begin{figure}[t]
		\centering\includegraphics[trim=0.00in 0.07in 0.07in 0.1in, clip,width=0.50\textwidth]{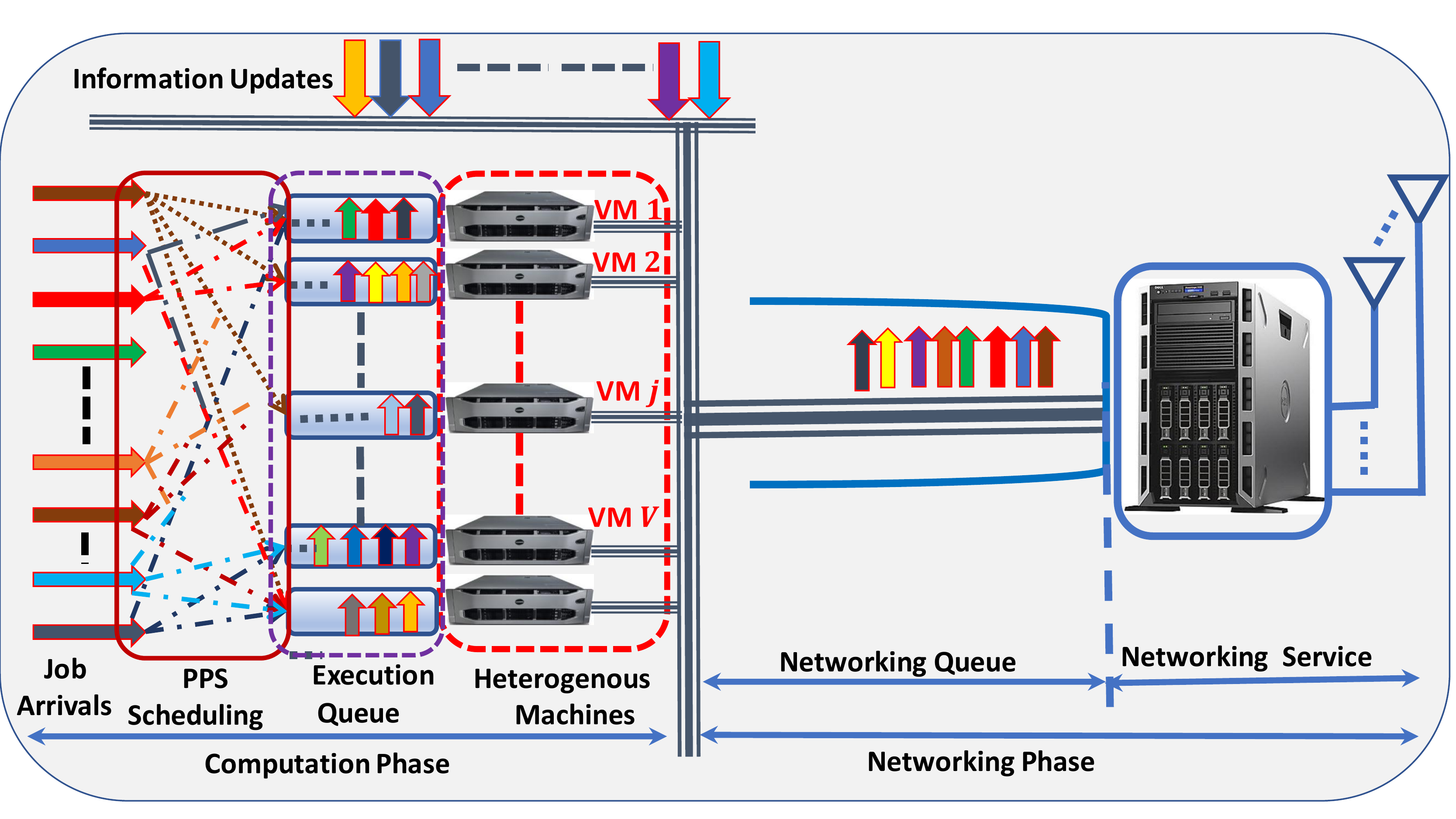}
		\caption{\textcolor{black}{System model diagram showing the different stages per phase. 
		\label{sysModel_zoomIn}}}
	\end{figure}
}
We consider an information-update system, composed of a single server equipped with multiple heterogeneous cores or machines (i.e., server farm\footnote{A server farm (or server cluster) is a set of computer nodes that are maintained by an organization to provide server functionality far beyond the capability of a single server.}). \textcolor{black}{ A typical schematic for a server farm is shown in Figure \ref{sysModel1}. This data center is composed of one edge-router connected to several top-of-rack switches, each of which is associated with multiple server racks. Without loss of generality, we consider, in the rest of this paper, only one TOR switch connected to multiple servers/cores or VMs (depicted by the red dashes in Figure \ref{sysModel1}). However, as show in Section \ref{ext}, our analysis can be easily extended to accommodate the scenarios where several hierarchical layers of routers and/or TOR switches are found. The system model and the different phases are captured in Figure \ref{sysModel_zoomIn}. As shown in the figure, every job passes through two phases: computation phase and networking phase. Further details on each phase will come later on in the sequel. }   
We assume that each vehicle $j$  requests a computation at a rate of $\lambda_j$ from the data center. Each job from vehicle $j$ requires information from multiple vehicles in the set ${\mathbf I}_j\subseteq {\mathbb J}$. 


The jobs from these vehicles are sent by users to perform some tasks on the server. Let $\mathbb{V}=\left\{ 1,2,\ldots,V\right\} $ be the set of heterogeneous virtual machines (VMs) in the server farm, where $v$ denotes a VM in the set $\mathbb{V}$, i.e., $v\in \mathbb{V}$.  Each job from vehicle $j$ has a fixed size $D_{j}$ so that it takes $D_j\times t_v$ units of time to complete, where $t_v$ is the time taken if assigned to machine $v$ for a unit-sized job. \textcolor{black}{
We note that $D_j$ has a heavy tail and follows a Pareto distribution with parameters $u_m,\gamma$ with shape parameter $\gamma>1$, implying finite mean and variance. Thus, the complementary distribution function of $D_j$ is given as 
\begin{equation}
	\mathbb{P}\left(D_{j}>u\right)\leq\begin{cases}
	(u_{m}/u)^{\gamma} & u\geq u_{m}.\\
	0 & u<u_{m}.
	\end{cases}
	\end{equation} 
	For $\gamma$, the mean is $\mathbb{E}\left[D_{j}\right]=\gamma\,u_{m}/(\gamma-1)$.
}
The job sizes reflects the time to perform the computation of the job on a server. Higher the job size, higher is the computational time.  The jobs can be assigned to any VM $v$ to process. Further, jobs are assumed to be non-preemptive so jobs cannot be interrupted 
if they are already in service. 
In order to serve a request from a vehicle, we first need to choose one VM $v$ to perform the computation task. The selection process of VMs is a challenging task as it needs to take into consideration many factors including the queue of each VM as well as the running jobs that are not fully executed yet. In Section \ref{schedComp}, we will explain our proposed scheduling policy and will show how this policy is optimized to reduce the completion times of the jobs. Figure \ref{timeLinejobj} shows the two stages of computing phase. As shown in the figure, a job $j$ waits some time $(W_{1,j,v})$ in the queue of VM $v$ before it is being served. Then, it spends some time in service, denoted as $S_{1,j,v}$. Once a job $j$ has been fully served, it is then directed to the  networking phase where the output of the computing should be routed to its destination.  
%
%

We also assume that the computational (service) time of a job $j$ at VM $v$ follows a shifted exponential distribution. This can be seen for instance from the experiments in \cite{7274674}, where the job completion time plots can be seen to follow a shifted exponential distribution. The shift represents that each job takes large time to complete, while the exponential part is motivated by the randomness in the background and other processes that make the run time non-deterministic.  Thus, the distribution of computation time for a  job of vehicle $j$ is given by 
\begin{equation}
f_{v,j}(x)=\begin{cases}
\alpha_{v,j}e^{-\alpha_{v,j}\left(x-\beta_{v,j}\right)} & x\geq\beta_{v,j}\\
0 & x<\beta_{v,j} \label{sExp}
\end{cases}
\end{equation}
where $\alpha_{v,j} = \alpha_v/{ D}_j$ and $\beta_{j,v} = \beta_v { D}_j$. Further, the expected processing time of a job $j$ is ${ D}_j (\beta_v + 1/\alpha_v)$, $\beta_v$ represents the shift and $\alpha_v$ represents the rate of the exponential random part.


\subsection{Networking Model}

When a job $j$ departs from the VM $v$ after execution, it enters the second (networking) phase where all completed jobs are placed into a common queue. This phase is called the networking phase where each job is queued for service in order to be sent back to the target vehicle. This phase also has two components: waiting for service in the queue $W_{2,j}$ and service time $S_{2,j}$. Figure \ref{timeLinejobj} depicts the two stages of the networking phase of any job $j$. 

We  assume that the jobs are non-preemptive so if a job is already on service it cannot be interrupted and will last running till it completes. We will assume that different jobs have different weights, thus FCFS is not an optimal scheme for sending the results of computation to the users.  We will later show (Lemma \ref{theo1}) that for ideal computing phase, a priority queuing will be optimal for  the metric which weights the AoI and completion time of jobs. 


We assume that the output of jobs from vehicle $j$ after computation require downloading ${\mathbf E}_j$ amount of data. We also assume that the networking service time of a job $j$ follows a shifted exponential distribution, as it has  been shown in real experiments, see for instance Tahoe system \cite{xiang2016joint} and Amazon S3 \cite{chen2014queueing}. We assume that the distribution of networking time for a  job of vehicle $j$ is given by 
\begin{equation}
f_{j}(x)=\begin{cases}
\gamma_{j}e^{-\gamma_{j}\left(x-\zeta_{j}\right)} & x\geq\zeta_{j}\\
0 & x<\zeta_{j} \label{sExp2}
\end{cases}
\end{equation}
where $\zeta_j = \zeta{\bf E}_j$ and $\gamma_{j}=\gamma/{\bf E}_j$ are the shift and the rate for the shifted exponential distribution, respectively.
%

\textcolor{black}{
	\begin{figure*}[t]
		\centering\includegraphics[trim=0.00in 0.0in 0.00in 0.0in, clip,width=0.81\textwidth]{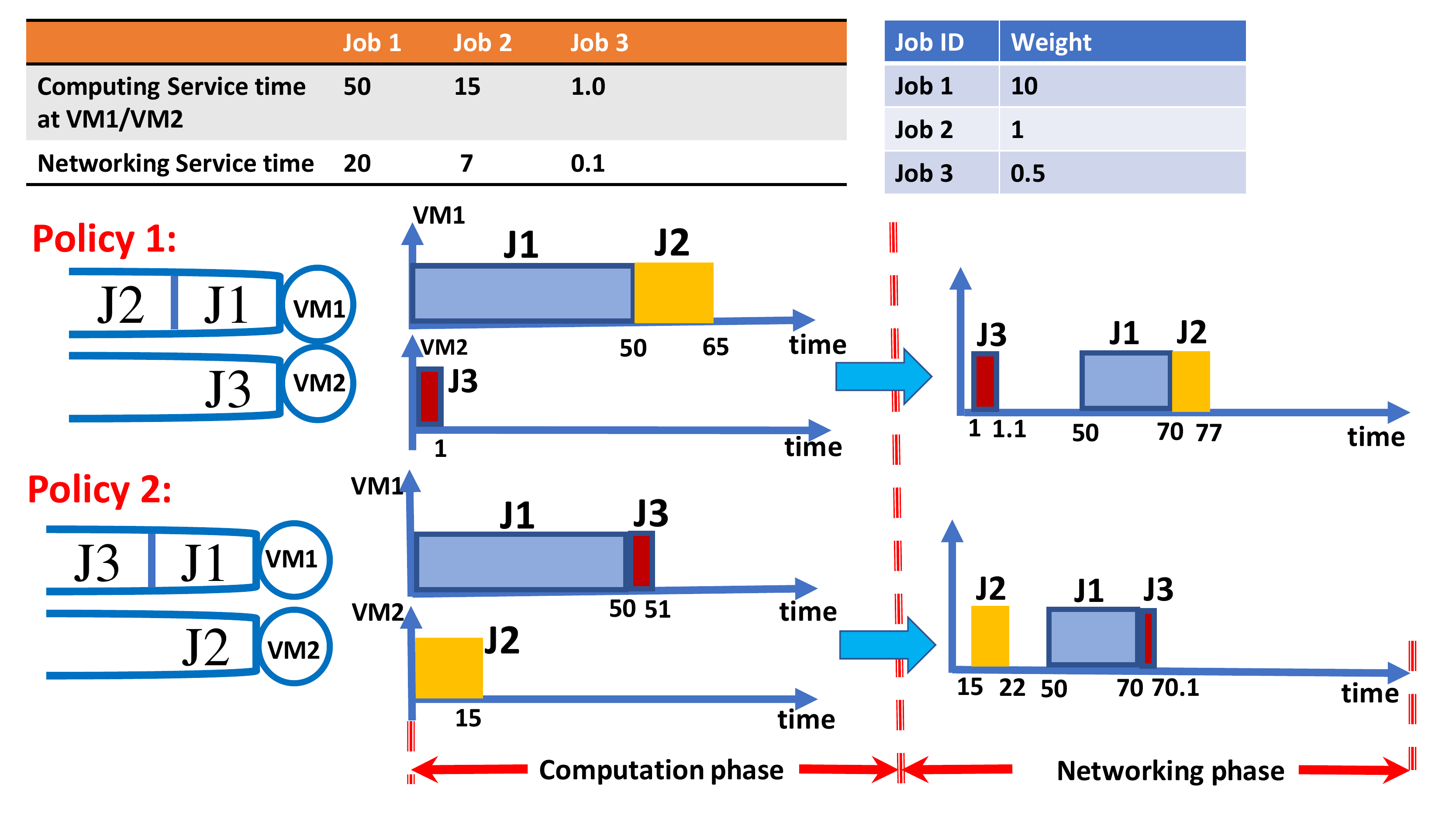}
		\caption{\textcolor{black}{Example illustrates the tradeoff between AoI and computation time metrics. There are three jobs ($J1$, $J2$, and $J3$) and two VMs (VM1 and VM2) for computation, followed by a common queue for networking. The computation times, $C_J$'s, for jobs, are shown in the figure and $C_j$'s are assumed to be equal on both machines. The two phases (computation and networking) and the timelines of the jobs for the two considered policies are depicted in the figure above.  
				\label{example}}}
	\end{figure*}
}

\subsection{Problem Formulation}

Having defined the system model, we now define the notions of completion time and the AoI metric. 

We note that the completion time for a job depends on the four stages, (i) Waiting for computation in the queue of assigned VM $v$, ${\mathbf W}_{1,j,v}$, (ii) Computation of the job from VM $v$, ${\mathbf S}_{1,j,v}$, (iii) Waiting in the networking queue for transmission to the end user, ${\mathbf W}_{2,j}$, and (iv) Service time for networking, ${\mathbf S}_{2,j}$. The completion time for a job from vehicle $j$ is thus given as 
\begin{equation}
{\bf C}_j = {\mathbf W}_{1,j,v}+ {\mathbf S}_{1,j,v} + {\mathbf W}_{2,j} + {\mathbf S}_{2,j}.
\end{equation}

Further, the AoI metric is determined by the update process. We note that the updates can keep happening before the computation starts. Thus, the AoI does not depend on the time taken in the queue of the VM. Let $c_j$ be the time that computation of job $j$ starts and $l_i\le c_j$ be the last time that an update from vehicle $i$ is received before $c_j$. The AoI is thus composed of four phases, (i)  $c_j - \max_{i\in {\bf I}_j}l_i$, which is the time of the last update of the most stale vehicle information required, ${\mathbf Y}_j$, (ii) Computation of the job from VM $v$, ${\mathbf S}_{1,j,v}$, (iii) Waiting in the networking queue for transmission to the end user, ${\mathbf W}_{2,j}$, and (iv) Service time for networking, ${\mathbf S}_{2,j}$. The AoI  for a job from vehicle $j$ is thus given as 
\begin{equation}
{\bf A}_j = {\mathbf Y}_{j}+{\mathbf S}_{1,j,v} + {\mathbf W}_{2,j} + {\mathbf S}_{2,j}. \label{ageEqn}
\end{equation}

We note that ${\mathbf E}[{\mathbf Y}_{j}]$ only depends on $\mu_j$'s and is a constant (independent of any scheduling strategy). Due to exponential updates, the distribution of $[{\mathbf Y}_{j}]$ is the maximum of the update distributions for each vehicle in ${\bf I}_j$. Thus, we ignore ${\mathbf Y}_{j}$ in the computations. 

\textcolor{black}{
{\it Age and completion time tradeoff:}
Figure \ref{example} shows two different policies where Policy 1 gives lower weighted age, while Policy 2 gives lower weighted completion time. In this example, we assume two VMs and three different jobs whose service times for computation are $50$ sec, $15$ sec, and $1$ sec, for job 1, job 2, and job 3, respectively. Further, the networking service times for job 1, job 2, and job 3 are, respectively, $20$ sec, $7$ sec, $0.1$ sec. We assume that the service times are deterministic for this example illustration, which is a special case of shifted exponential distribution assumed in the paper.  As shown in the figure, using policy 1 (policy 2), the completion times for jobs 1, 2, and 3, are $70$ ($70$), $77$ ($22$), and $1.1$ ($70.1$), respectively. Using equation \eqref{ageEqn}, we can calculate the age ($A_i$ for job $i$), for every job. Thus, for policy 1, we have $A_{1}=70$, $A_{2}=27$, and $A_{3}=1.1$, while the ages of jobs using policy 2 are $A_{1}=70$, $A_{2}=22$, and $A_{3}=20.1$. Note that the completion time and age for job $1$ are the same in both policies and thus we ignore it in the following comparison. Our objective is to minimize a weighted sum of ages and completion times of jobs, $\sum_j(w_j {\bf C}_j + g_j {\bf A}_j)$. Hence, for policy 1, assuming $g_i=w_i$ and $j\in \{1,2\}$,  we have $\sum_j(g_j {\bf A}_j)=27.550$, and that for policy 2, we can show that $\sum_j(g_j {\bf A}_j)=32.05$. On the other hand, the weighted completion time for policy 1 is $\sum_j(w_j {\bf C}_j) = 77.55$ and similarly for policy 2, we have $\sum_j(w_j {\bf C}_j) =57.05$. We can see that Policy 1 achieves lower weighted age while Policy 2 obtains lower weighted completion time. Therefore, based on the chosen policy we get different set of outcomes. Hence, an efficient tradeoff point between age and completion time can be designed based on the level desired by the application such as the tolerable level of staleness in the delivered information. 
}

Having defined the completion time and the AoI metric, we aim to minimize the weighted completion and AoI metric. Thus, the problem formulation is to minimize $\sum_j(w_j {\bf C}_j + g_j {\bf A}_j)$. Ignoring  ${\mathbf Y}_{j}$, the above reduces to minimizing $\sum_j \left( w_j{\mathbf W}_{1,j,v}+ (w_j+g_j){\mathbf S}_{1,j,v} + {\mathbf W}_{2,j} + {\mathbf S}_{2,j}\right)$.


\section{Proposed Scheduling Approach}

In this section, we explain our proposed strategies for job scheduling both for the computation and networking stages. 

\subsection{Scheduling for Computing}
\label{schedComp}

We now describe a feasible job scheduling policy that would have parameters that can be used to optimize the proposed metric. Upon arrival of jobs at the server, a VM has to be chosen to perform the task. The optimal scheduling policy has to consider the queue state and all tasks that are not executed yet. While one can use a Markov decision process with multiple states, this approach is not tractable and will result in, so-called, state explosion problem \cite{xiang2016joint}. Further, this approach will not give formulas that can be optimized to determine the job assignments and optimal resource allocation of the server farm. To overcome these issues, we propose feasible scheduling to jointly consider all different design parameters. Hence, to provide prioritized service levels, we propose a prioritized probabilistic policy (PPS) as follows. Each VM $v$ has its own queue and the jobs in each queue are served under First Come First Serve. A job $j$, $j=\{1,2,3,\ldots,J\}$, is assigned to the queue of a VM $v$, $v=\{1,2,\ldots,V\}$, with probability $p_{j,v}\geq0$ and for any job $j$, the following condition has to be satisfied for feasibility of scheduling process
\begin{equation}
\sum_{v=1}^{V}p_{j,v}=1,\,\,\forall j. \label{sum_p}
\end{equation}

In order to serve a request from a vehicle, we first assign probabilistically one VM to perform the  computation task.  VM $v$ is assigned to serve a job $j$ with probability $p_{j,v}$. Since the key bottleneck is the number of VMs,  jobs have to wait in the queue until the VM is free and then can serve them. Thus, if the VMs are busy computing other tasks, the incoming
requests have to wait in the queues. Moreover, we assume that requests at the queue of each VM are served in order of the requests in a FIFO fashion. Under probabilistic scheduling, the arrival of job requests at VM $v$ forms a Poisson process with rate $\Lambda_v= \sum_{j} p_{j,v} \lambda_j$, which is the superposition of $J$ Poisson processes each with rate $p_{j,v} \lambda_j$. We present next the proposed scheduling for the networking phase. We note that the AoI depends only on the computation service time, while the completion time depends on both the waiting time in the queue of VM as well as the service time. 
\subsection{Scheduling for Networking}
\label{netSched}
When a job from vehicle $j$ departs  the VM $v$ after completion, it enters another queue waiting to be sent back to the target vehicle. To further illustrate, each job has two phases of service. The first phase is the computation phase and the second phase is the networking phase. In the first phase, the jobs are assigned to the VMs probabilistically  to run the tasks as described earlier. 
Once the  job from vehicle $j$ finishes the first phase, it enters the networking phase where all completed jobs are placed into a common queue. This phase  has two parts: waiting for service in the queue and service time. Figure \ref{timeLinejobj} depicts the different phases/stages of processing time of any job $j$. Since both the AoI and the completion time depend on both the parts, the weights of the two can be combined for completion of networking part. For weighted completion of jobs, FCFS (first come first serve) scheduling may not be optimal. We thus use a priority queuing, which will be shown to be optimal for the networking part (Lemma  \ref{theo1}).

%
%
\begin{figure}[t]
	\centering\includegraphics[trim=1.35in 0.15in 3.5in 0.51in, clip,width=0.491\textwidth]{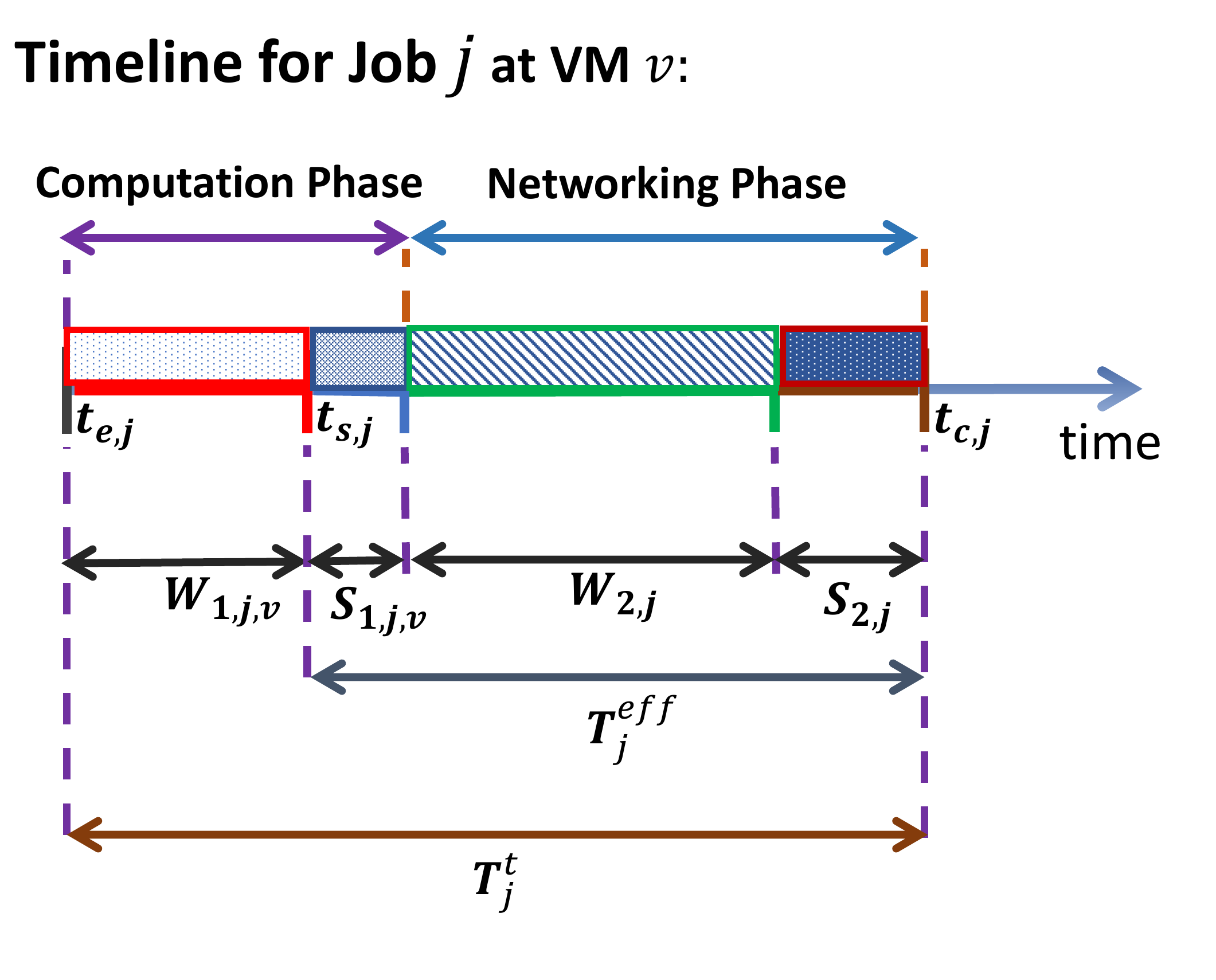}
	\caption{ A schematic illustrates the timeline of a job $j$ at VM $v$.  The computing and networking phases are captured. Further, the different stages per each phase is also shown. Here, $t_{e,j}$, $t_{s,j}$, and $t_{c,j}$ represent the time at which a job $j$ enters the computation queue of VM $v$, starts the compute service, and completes the networking service time, respectively. Further, ${\bf T}_j^{eff}$  corresponds to the effective time that contributes to the AoI and ${\bf T}_j^t$ is the total processing time (computing and networking) of a job $j$ in VM $v$.
		\label{timeLinejobj}}
\end{figure}


%
%
%
%
%
%
%
%
%
%
%
%
%

%


\section{AoI Analysis}
\label{AoI_analysis}
 In this section, we derive expressions for both the completion time and the AoI metric. The expressions require characterizing the terms ${\mathbf W}_{1,j,v}$, ${\mathbf S}_{1,j,v}$, ${\mathbf W}_{2,j}$, and ${\mathbf S}_{2,j}$. We will split computing the first two terms for the computation stage and the next two for the networking stage in the two subsections. This will be followed by combining them to give expressions for average completion time and mean AoI in the following subsection.

 \subsection{Evaluating terms in the computing phase}

For every VM $v$ and based on the PPS policy, this phase of the system can be modeled using an M/G/1 queuing model since the arrival of jobs is Poisson (with different
probabilities to each VM queue) and the execution time is general (i.e., shifted exponential) related to the processing time of each job $j$. Recall that the overall arrival rate of the jobs at VM $v\in \mathbb{V}$ is 
\begin{equation}
\Lambda_v= \sum_{j} p_{j,v} \lambda_j \label{arrRate_v}
\end{equation}

Given that the service time distribution is shifted exponential whose expression is given in \eqref{sExp},  it is easy to show that for every job from vehicle $j$ the expected service time at VM $v$ is $\left(\beta_{v,j}+\frac{1}{\alpha_{v,j}}\right)
$. Then, following \cite{ross2014introduction}, the expected service time at VM $v$ is given by 
\begin{equation}
\mathbb{E}\left[\boldsymbol{Z}_{1,v}\right]=\sum_{j=1}^{J}\frac{p_{j,v}\lambda_{j}}{\Lambda_{v}}\left(\beta_{v,j}+\frac{1}{\alpha_{v,j}}\right)
\end{equation}

Similarly, we can calculate the second moment of the expected service time as follows

\begin{equation}
\mathbb{E}\left[\boldsymbol{Z}_{1,v}^{2}\right]=\sum_{j=1}^{J}\frac{p_{j,v}\lambda_{j}}{\Lambda_{v}}\left(\beta_{v,j}^{2}+\beta_{v,j}+\frac{\beta_{v,j}+2}{\alpha_{v,j}}\right) \label{serSeconMmt}
\end{equation}

Since the arrivals of jobs from vehicle $j\in \mathbb{J}$ are Poisson and the service time is general (i.e., M/G/1), the expected waiting time in the queue for any VM $v\in \mathbb{V}$ is given by 
\begin{equation}
\mathbb{E}\left[\boldsymbol{W}_{1,j,v}\right]=\frac{\Lambda_{v}\mathbb{E}\left[\boldsymbol{Z}_{1,v}^{2}\right]}{2\left(1-\Lambda_{v}\mathbb{E}\left[\boldsymbol{Z}_{1,v}\right]\right)}
\end{equation}

Further, in order for the queuing system to be stable, the arrival rate of jobs from vehicle $j$ at VM $v$ has to be less than its expected service rate, which implies that
\begin{equation}
\Lambda_{v}\mathbb{E}\left[\boldsymbol{Z}_{1,v}\right]=\rho_{1,v}<1,\,\,\forall v \label{rholeq1}
\end{equation}

Further, the expected service time of job of vehicle $j$ from VM $v$ is given as
\begin{equation}
\mathbb{E}\left[\boldsymbol{S}_{1,j,v}\right]= \left(\beta_{v,j}+\frac{1}{\alpha_{v,j}}\right)
\end{equation}

To this end, we have characterized the waiting time and the service time for a job from vehicle $j$ in any given VM $v$ under PPS policy. Then, the average computing (waiting and service) time for a job $j$ in a VM $v$, $\mathbb{E}[\boldsymbol{T}_{1,j,v}^{(c)}]$, is given by 
\begin{equation}
\mathbb{E}[\boldsymbol{T}_{1,j,v}^{(c)}]=\mathbb{E}\left[\boldsymbol{W}_{1,j,v}\right]+\mathbb{E}\left[\boldsymbol{S}_{1,j,v}\right].
\end{equation}

 \subsection{Evaluating terms in the networking phase}

Having characterized the queuing dynamics of the execution phase, the distribution of the arrivals at the queue of the second phase can be found as follows. Since each VM has input process as Poisson and the service time is shifted exponential, the output process from the queue is Poisson. We note that this result does not hold for general service time distributions. However, it has been shown for exponential service time distribution \cite{ross2014introduction}, and it easily follows for shift distribution. Thus, the result follows for the shifted exponential service times. Since superposition of Poisson processes from different VMs is also a Poisson, the arrivals at the second queue is also Poisson with rate $\Lambda = \sum_{v} \Lambda_v = \sum_j \lambda_j$. At this stage, the computation result of the first phase is available and is ready  for transmission. 


We assume that the vehicles are ordered in terms of decreasing $(w_j+g_j)/{\mathbf E}_j$. For minimizing completion time for the networking, it is optimal to have strict priority in the jobs, where the jobs from lower vehicle number have strict priority over the higher ones. Thus, this leads to an M/G/1 system with priority queuing. Thus, the expected waiting time in the queue for a job of vehicle $j$ is given as \cite{pinedo2016scheduling}

\begin{equation}
\mathbb{E}\left[\boldsymbol{W}_{2,j}\right]=\frac{\Lambda\mathbb{E}\left[\boldsymbol{Z}_{2}^2\right]}{2\left(1-\sum_{z=1}^{j}\rho_{z}\right)\left(1-\sum_{z=1}^{j-1}\rho_{z}\right)}
\end{equation}
where $\mathbb{\rho}_{z}=\lambda_{z}\mathbb{E}\left[\boldsymbol{S}_{2,z}\right]
, z \in \mathbb{J}$ is the traffic intensity, and $\mathbb{E}\left[\boldsymbol{Z}^2_{2}\right]
$ denote the second moment of the service time for networking service time. These expressions for $\mathbb{E}\left[\boldsymbol{S}_{2,j}\right]$ and $\mathbb{E}\left[\boldsymbol{Z}^{2}_{2}\right] $ are given as

\begin{align}
\mathbb{E}\left[\boldsymbol{S}_{2,j}\right] & =\left(\zeta_{j}+\frac{1}{\gamma_{j}}\right)\label{serExp2ndQueue}
\end{align}

\begin{align}
\mathbb{E}\left[\boldsymbol{Z}^{2}_{2}\right] & =\sum_j \frac{\lambda_{j}}{\Lambda}\left(\zeta_{j}^{2}+\zeta_{j}+\frac{\zeta_{j}+2}{\gamma_{j}}\right) \label{eq:ser2ndQueue}
\end{align}

We note that for stability of the queuing system at the second phase, the arrival rate of jobs $j$'s has to be less than its expected service rate, which implies that
\begin{equation}
\Lambda\sum_j \frac{\lambda_j}{\Lambda}\mathbb{E}\left[\boldsymbol{S}_{2,j}\right]<1 \label{rholeq1_2ndQ}
\end{equation}


\subsection{Expressions for Average Completion Time and Mean AoI}
Now, we are ready to state the main results of this paper. The following theorem proves that the proposed priority (given higher priority for jobs of larger $((w_j+g_j)/{\bf E}_{j})$) will minimize both completion time and AoI for the networking phase.

\begin{lemma}
	\label{theo1}
Ignoring the computation phase, priority queue scheduling which gives a strict priority to the jobs with higher  $(w_j+g_j)/\boldsymbol{E}_j$ optimizes the metric in this paper. Thus,  the jobs are scheduled according to the weighted shortest expected processing time (WSEPT) rule.	
\end{lemma}

\begin{proof}
The proof of this Lemma relies on the concept of Adjacent Pairwise Interchange (API) on jobs. Since the proof of this theorem follows directly from the proof of Theorem 11.3.1 in \cite{pinedo2016scheduling}, we refer the interested reader to the book in \cite{pinedo2016scheduling}, page $301$, for a detailed treatment of this.	
\end{proof}


Using conditional expectation over the choice of VM $v$, the expected AoI and completion time for any job $j\in \{\mathbb{J}\}$ can be expressed as in the following theorems.

\begin{theorem} \label{AoI_j}
The expected AoI for any job $j\in \mathbb{J}$ is given by
\begin{equation}
\mathbb{E}\left[\boldsymbol{A}_{j}\right]=\sum_{v=1}^{V}p_{j,v}\mathbb{E}\left[\boldsymbol{S}_{1,j,v}\right]+\frac{\lambda_{j}}{\Lambda}\left(\mathbb{E}\left[\boldsymbol{W}_{2,j}\right]+\mathbb{E}\left[\boldsymbol{S}_{2,j}\right]\right)\label{AoI_{j}2}
\end{equation}

\end{theorem}

\begin{proof}
	
We start by noting that the waiting time
of job $j$ in the computing phase before processing does not count to the AoI since the most recent information will be used once a task has been placed for execution. Thus, only the computing service time, waiting time in the networking queue, and service time of networking will contribute to the AoI. Since the age of job $j$ is $1/\mu_j$, the age increases at its current age $1/\mu_j$ by the amount of processing times in the computing service stage and the additional time incurred by networking phase. By averaging over the choice of VM, the statement of the theorem in  \eqref{AoI_{j}2} follows.   
\end{proof}

\begin{theorem} \label{C_j}
	The expected completion time for any job $j\in \mathbb{J}$ is given by
\begin{align}
\mathbb{E}\left[\boldsymbol{C}_{j}\right]=\sum_{v=1}^{V}\left[p_{j,v}\left(\mathbb{E}\left[\boldsymbol{W}_{1,j,v}\right]+\mathbb{E}\left[\boldsymbol{S}_{1,j,v}\right]\right)\right] \nonumber \\
+\mathbb{E}\left[\boldsymbol{W}_{2,j}\right]+\mathbb{E}\left[\boldsymbol{S}_{2,j}\right]	\label{C_{j}2}
\end{align}

\end{theorem}

\begin{proof}
	
We note that the expected release (arrival) time of a job $j$ is $\frac{1}{\lambda_j}$. This time also represents the earliest time, on expectation, that a job $j$ can be available to start processing. Clearly, the amount of work still to be done for a job $j$ to be complete increases at a release time by the amount of processing times in both computing phase and networking phase. By averaging over the choice of the VM and ignoring the term $(1/\lambda_{j})$ since its is fixed and thus cannot be optimized, we can find the expected completion time of a job $j$ as given in \eqref{C_{j}2}
which proves the statement of the theorem. 
\end{proof}

\section{Tradeoff between AoI and Completion time}
In this section, we will first formulate a joint AoI and completion time optimization for multiple heterogeneous jobs which determines the optimal PPS scheduling probabilities. The formulation is then shown to be a convex optimization problem. The solution is then used to develop an algorithm for scheduling at the compute server. 

\subsection{Joint Optimization for AoI and Completion time}
\label{proForm}
Now we formulate a joint AoI and completion time optimization for multiple heterogeneous jobs. Let $\boldsymbol{p}=(p_{j,v},\,\forall j=\{1,2,\cdots,N\},\,\text{and}\, v=\{1,2,\cdots,V\})$. 
Our goal is  to minimize both AoI and completion times over
the choice of scheduling access decisions $\boldsymbol{p}$. Since this is 
a multi-objective optimization, the objective can be modeled
as a convex combination of the two metrics. Let $\lambda=\sum_{j} \lambda_j$ be the total arrival rate of jobs.
 Then, $w_j=\lambda_j\theta/\lambda $ is the ratio of job $j$ requests
   to the total requests of all jobs, where $\theta \in [0,1]$ is a trade-off factor that determines the relative significance of completion time and AoI in the optimization problem.  The first objective is the minimization of the total completion time, averaged over all job requests, and is given by $\sum_{j}w_j\bf{E}\left[\bf{C}_{j}\right]$. The second objective is the minimization of AoI of all jobs, averaged over all job requests, and is given by $\sum_{j}g_j\bf{E}\left[\bf{A}_{j}\right]$, where $g_j=\frac{\lambda_{j}}{\lambda}(1-\theta)$. 
   We note that networking part is already optimized by priority scheduling. The choice of $\boldsymbol{p}$ would only impact the computation phase, yet this choice still affects both AoI and completion time. Thus, we wish to find probabilities such that the computation phase is optimized for the weighted metric.
   Then, optimizing a convex combination of the two metrics can be formulated as follows.
\begin{align}
\text{\textbf{min}}\,\,\,\,\,\, & \sum_{j=1}^{N}\left[\theta \frac{\lambda_{j}}{\lambda}\mathbb{E}\left[\boldsymbol{C}_{j}\right]+(1-\theta)\frac{\lambda_{j}}{\lambda}\mathbb{E}\left[\boldsymbol{A}_{j}\right]\right] \label{optPob}\\
\text{\textbf{s.t.}}\,\,\,\,\,\,\nonumber \\
& \eqref{sum_p},  
\eqref{arrRate_v}, 
\eqref{rholeq1}, 
\eqref{rholeq1_2ndQ}, 
\eqref{AoI_{j}2},
\eqref{C_{j}2}, \label{const1}\\
& p_{j,v}\geq0,\,\,\forall j,v \label{pgeq0}\\
\text{\textbf{var}}\,\,\,\,\,\, &\,\,\, \boldsymbol{p}\nonumber 
\end{align}

Here $\theta \in [0,1]$ is a trade-off factor that determines the relative significance of completion time and AoI in the optimization problem. 
Varying $\theta=0$ to $\theta=1$, the solution for \eqref{optPob} spans the solutions that minimize the completion time to ones that minimize the AoI of jobs. Constraint \eqref{sum_p} ensures the feasibility of the scheduling probabilities. Further, Constraint \eqref{arrRate_v} gives the aggregate arrival rate $\Lambda_{v}$ for each VM $v$ given the PPS and the arrival rates of jobs $\lambda_j$'s, while Constraint \eqref{rholeq1} ensures that the load intensity at VM $v$ is less than one for stable system. Constraint \eqref{rholeq1_2ndQ} gives the load intensity at the second queue. Constraints  \eqref{AoI_{j}2} and 
\eqref{C_{j}2} give the expression for the expected completion time and expected AoI for all jobs $j$, respectively. Finally, Constraint \eqref{pgeq0} ensures the positivity of the scheduling probabilities. Next, we present our proposed algorithm to solve the optimization problem. 

\subsection{Convexity of finding PPS scheduling probabilities}

In this subsection, we 
show that the optimization problem is convex for all $p_{j,v}$, $j\in \mathbb{J}$ and $v\in \mathbb{V}$.

\begin{theorem} \label{convxOptProb}
The objective function defined in \eqref{optPob} is convex in $\boldsymbol{p}=(p_{j,v},\,j=1,2,3,\ldots,J,\, and\,v=1,2,3,\ldots,V)$ in the region where constraints in \eqref{const1}-\eqref{pgeq0} are satisfied. 

\end{theorem}

\begin{proof}
Since the sum of convex functions is convex \cite{boyd2004convex}, it is enough to show that:
\begin{equation}
\sum_{v}\Lambda_{v}\mathbb{E}\left[\boldsymbol{W}_{1,j,v}\right]+\sum_{v}\Lambda_{v}\mathbb{E}\left[\boldsymbol{S}_{1,j,v}\right]+  \sum_{j}\frac{\lambda j}{\lambda}\mathbb{E}\left[\boldsymbol{W}_{2,j}\right]\label{eqConvx}
\end{equation}
is convex. However, the last term is independent on $\boldsymbol{p}$ and thus can be ignored. Further, since  $\Lambda_{v}$ is a linear function of $\boldsymbol{p}$, it is enough to prove the convexity with respect to $\Lambda_{v}$. It is easy to show that the second term  in \eqref{eqConvx} is linear functions in $\Lambda_{v}$ and thus convex. For the first term in \eqref{eqConvx}, after simplifying the term, it can be written as
\begin{align}
\frac{\sum_{j}p_{j,v}\lambda_{j}\left(\beta_{v,j}^{2}+\beta_{v,j}+\frac{\beta_{v,j}+2}{\alpha_{v,j}}\right)}{1-\sum_{j}p_{j,v}\lambda_{j}\left(\beta_{v,j}+\frac{1}{\alpha_{v,j}}\right)} \label{1stTerm}
\end{align}
%
Now,  we need to show that the above function is convex in $\boldsymbol{p}$. We note that the above function can be expressed as $\eta\boldsymbol{x}^{T}(\boldsymbol{I}-\boldsymbol{x})^{-1}$,
where $\eta>1$ and $\boldsymbol{x}$ is a $J$-dimensional vector. The role of $\eta$ is to capture the ratio between the two terms $\sum_{j}p_{j,v}\lambda_{j}\left(\beta_{v,j}^{2}+\beta_{v,j}+\frac{\beta_{v,j}+2}{\alpha_{v,j}}\right)$
and $\left(\sum_{j}p_{j,v}\lambda_{j}\left(\beta_{v,j}+\frac{1}{\alpha_{v,j}}\right)\right)$.
This function is jointly convex on $x_{j}$ for all $j$ if and only
if $x_{j}<1$, for all $j$, see \cite{boyd2004convex} for detailed treatment on this. In our expression, $x_{j}$ represents the load
intensity (i.e., $p_{j,v}\lambda_{j}(\beta_{j}+\frac{1}{\alpha_{j}})$)
and should be less than one for stability of the system. Thus, the
above expression is convex with respect to $\boldsymbol{p}$.
\end{proof}

\subsection{Proposed Algorithm}\label{prop_algo}

In this subsection, we explain our proposed algorithm to optimize age and completion time tradeoff. The algorithm is summarized as follows. 
Using the arrival rates $\lambda_j$'s and a desirable trade-off factor $\theta$, the solution for the optimization problem in \eqref{optPob} gives the optimal PPS scheduling probabilities, i.e., $\boldsymbol{p}^{*}$. Since the Optimization problem is proven to be convex for $\boldsymbol{p}$, it can be
solved optimally by Projected Gradient Descent Algorithm, or any other convex solver \cite{boyd2004convex}.
When the computation part is completed, job $j$ is placed on the second queue (networking phase) for sending back to the vehicle.  In the networking phase, jobs are ordered according to the WSEPT rule.

\section{Extensions}
\label{ext}
In this section, we present two possible extensions for our proposed framework. First, we explain how our analysis can be extended to accommodate scenarios where several TORs switches (or edge-routers) are considered. Second, we develop an online algorithm for minimizing the AoI. 
\subsection{$m$-TORs Switches Case}
In this setting, unlike the previous case where only one TOR switch is considered, we consider multiple TOR switches. Thus, upon an arrival for a job request $j$, we first need to choose one of the TORs switches and then one of the servers or VMs. 
Hence, we extend the PPS scheduling proposed above  into two-stage PPS scheduling. The two-stage probabilistic policy chooses one TOR switches ($1$-out-of-$m$) and then, chooses $1$-out-of-$V$ VMs with certain probability, for every job $j$. Let $q_{j,u,v}$ be the probability of requesting job $j$ from the VM $v$
that belongs to the TOR switch $u$. Thus, $q_{j,u,v}$
is given by
\begin{equation}
q_{j,u,v}=\pi_{j,u}\,p_{u,v}\,,
\end{equation}
where $\pi_{j,u}$ is the probability of choosing TOR switch $u$ for job $j$, and $p_{u,v}$ is
the probability of choosing VM $v$ at TOR switch $u$.
The two-stage PSS  scheduling gives feasible probabilities for choosing $1$-out-of-$m$ TOR switches
and $1$-out-of-$V$ VMs if and only if there exists conditional probabilities $\pi_{j,u} \in \{0,1\}$ and $p_{u,v} \in \{0,1\}$ satisfying
\begin{align}
\sum_{u=1}^{m}\pi_{j,u} & =1\,\,,\forall j\,\\
\sum_{v=1}^{V}p_{u,v} & =1\,\,,\forall u\,
\end{align}
Using the two-stage PPS
scheduling,
the arrival of job requests at VM $v$ at TOR switch $u$ forms a
Poisson Process with rate $\Lambda_{u,v}=\sum_{j=1}^{N} \lambda_j q_{j,u,v}$ which
is the superposition of $N$ Poisson processes each with rate $\lambda_j q_{j,u,v}$. Then, by replacing $p_{j,v}$ by $q_{j,u,v}$ in Section \ref{AoI_analysis} (and in the following related expressions), we can drive the new completion time and AoI expressions in a similar fashion. Since this extension follows straightforward from our previous analysis, detailed derivation is omitted. 
\subsection{Online Algorithm for Age and Completion Time tradeoff Optimization }
While our proposed PPS scheduling policy is optimized for an offline scenario, an online algorithm can be derived according to the stationary PSS scheduling probabilities. The arrival rates $\lambda_j$ can be estimated based on a window based method. In this method, a window size $W$ is chosen, and the decisions in a window are based on the estimated arrival rates from the previous window. Using the estimated arrival rates, the solution for the optimization problem in \eqref{optPob} gives the optimal offline PPS scheduling probabilities, i.e., $\boldsymbol{p}^{*}$. According to these stationary scheduling probabilities, optimal randomized online policy can be obtained. Recall that this choice of $\boldsymbol{p}$ impacts only the computation phase, which in turn affects both AoI and completion time. Note that networking part is optimized by priority scheduling, i.e., jobs are scheduled according to the  WSEPT rule, which can be run in an online fashion. To summarize, the arrival rates of the files are estimated by a window-based method. Using the estimated arrival rates, the algorithm in Section \ref{prop_algo} is used. 


\section{Evaluation}
\label{simRes}

\begin{table}
	\caption{Values of $\alpha_{v}$ and $\beta_{v}$ used in the evaluation results
		with units of 1/ms \cite{xiang2016joint}.} 
	
	\centering%
	\begin{tabular}{|c|c|c|c|c|c|}
		\hline 
		Node & Node 1 & Node 2 & Node 3 & Node 4 & Node 5\tabularnewline
		\hline 
		\hline 
		$\alpha_{v}$ & 82 & 76 & 71 & 65 & 60\tabularnewline
		\hline 
		$\beta_{v}$ & $10$ & $12$ & $13$ & $17$ & $16$\tabularnewline
		\hline 
	\end{tabular}\\
	$\vphantom{}$
	
	$\vphantom{}$
	
	\centering%
	\begin{tabular}{|c|c|c|c|c|c|}
		\hline 
		Node & Node 6 & Node 7 & Node 8 & Node 9 & Node 10\tabularnewline
		\hline 
		\hline 
		$\alpha_{v}$ & 51 & 44 & 39 & 34 & 29\tabularnewline
		\hline 
		$\beta_{v}$ & $18$ & $20$ & $21$ & $23$ & $25$\tabularnewline
		\hline 
	\end{tabular}
		\label{VMparameters}
\end{table}

In this section, we evaluate our proposed algorithm for jointly optimizing the two proposed metrics of age and completion time. We employ a hybrid simulation method, where each machine has a different speed (depicted in Table \ref{VMparameters}), unless otherwise explicitly stated. Further, job sizes $D_j$ and $E_j$ are assumed to follow a heavy-tailed Pareto distribution \cite{paretoArnold015} as it is a commonly used
distribution for file sizes \cite{heavyTailedMor99}, with shape factor of 2 and scale
of 300, respectively. In our simulation, the job sizes are limited to be at most 5 times the mean. 
Unless otherwise stated, we set $J=1000$, and $V=10$. The arrival rate of job requests from vehicle $j$ is set to be $\lambda_b/(j+1)$, where $\lambda_b=1$, for all $j$. The rate $\gamma$ and shift $\zeta$ for the networking server is set to be $112$ /ms  and $18$ ms, respectively. While we stick in our simulation to these parameters, our analysis and results remain applicable for any
setting given that the system maintains stable conditions under the chosen parameters.
In order to initialize our algorithm, we assume uniform scheduling,  $p_{j,v}=1/V$. However, this choice of the scheduling probabilities may not be feasible. Thus, we modify this choice to be closest norm feasible solutions. 

\subsection{Comparisons}
The system performance of the developed joint optimization of the completion time and the AoI  is compared with two baseline systems described as follows:
\begin{enumerate}[leftmargin=0.25cm,itemindent=.5cm,labelwidth=\itemindent,labelsep=0cm,align=left]
	\item {\em Random Computing Assignments-Optimized Networking  (RCA-ON) Policy:} In this strategy, the jobs are assigned to VMs uniformly at random. Thus, the
	values of $p_{j,v}$ are set to be $1/V$. Further, scheduling of jobs for networking is optimized through prioritized queueing as explained in Section \ref{netSched}.
	
    \item {\em Optimized Computing Assignments-FCFS Networking  (OCA-FCFS) Policy:} In this strategy, job assignment for computing servers are optimized using the probabilistic scheduling policy, as explained in Section \ref{schedComp}. In addition, no priority is considered in the networking queue. Thus, the FCFS policy is assumed for placing the jobs after execution on the passing ports for transmission. Note that in this policy, we replaced the expression of the waiting time in the networking phase to that of the FCFS M/G/1 expression, i.e., $\mathbb{E}\left[\boldsymbol{W}_{2,j}\right]=\frac{\Lambda\mathbb{E}\left[\boldsymbol{Z}_{2}^{2}\right]}{2\left(1-\Lambda\sum_{j}\frac{\lambda_{j}}{\Lambda}\mathbb{E}\left[\boldsymbol{S}_{2,j}\right]\right)}
    $.
    
    \item {\em Proportional-service-rate Computing Assignments-Optimized Networking  (PCA-ON) Policy:}  The joint request scheduler
    chooses the access probabilities to be proportional to the
    service rates of the virtual machines, i.e., $p_{j,v}=(\beta_{v,j}+1/\alpha_{v,j})/\sum_{v}(\beta_{v,j}+1/\alpha_{v,j})$. This
    policy assigns VMs proportional to their service rates. These scheduling probabilities are projected toward feasible region to ensure stability of the system. Further, scheduling of jobs for networking is optimized through prioritized queueing as explained in Section \ref{netSched}.
\end{enumerate} 

\subsection{Results}
 The system performance is measured by different objectives including the overall weighted AoI, completion time as well as the tradeoff between these two metrics. We denote our proposed policy by PPS policy, where the computation time is optimized over the choice of the VM, i.e., optimizing the scheduling probabilities $(\boldsymbol{p})$, and the networking time is optimized through priority scheduling.
\textcolor{black}{
\begin{figure}[t]
	\centering\includegraphics[trim=0.1in 0.02in 2.89in 0.0in, clip,width=0.55\textwidth]{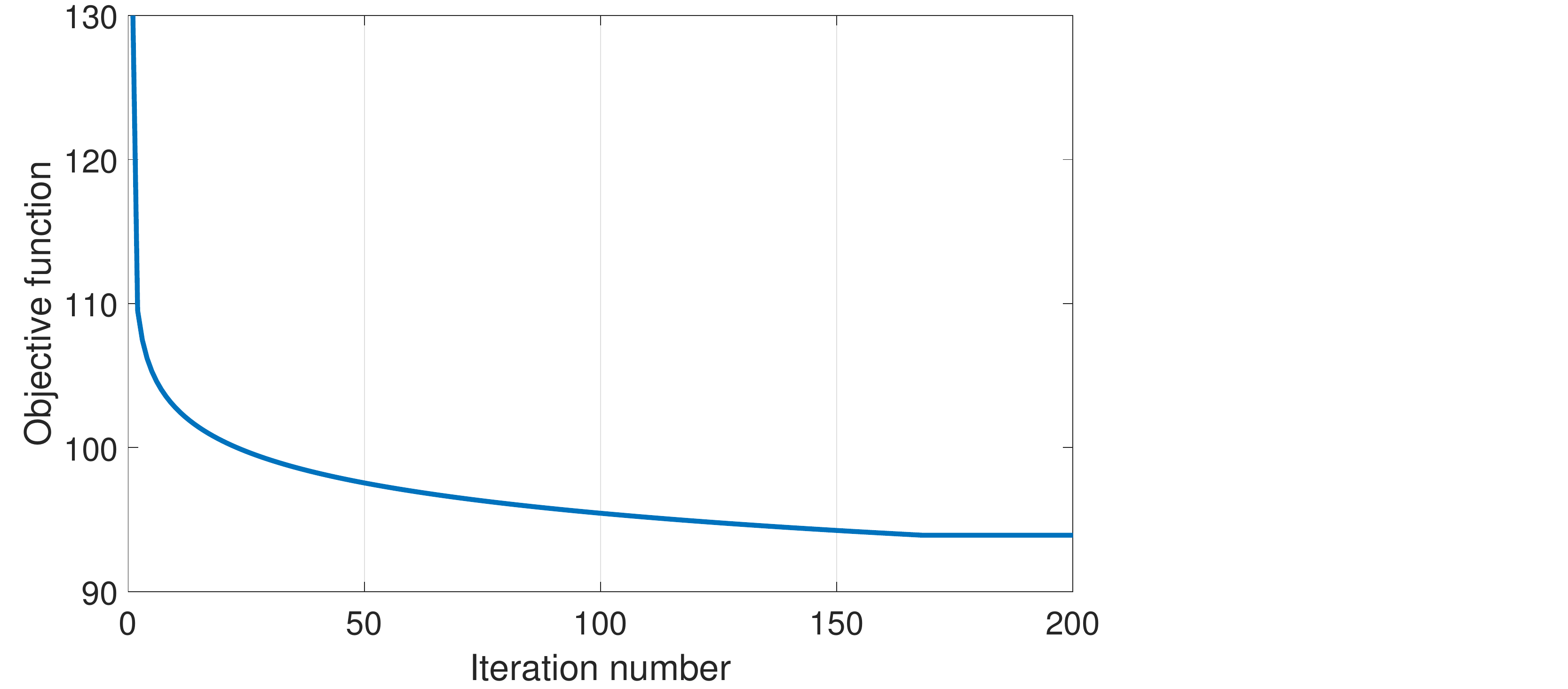}
	\caption{\textcolor{black}{Convergence of the objective function defined in \eqref{optPob}, a weighted sum of age and completion time,  versus number of iterations for $\theta=0.3$, $V=14$, and $N=1000$.  
		\label{conveg}}}
\end{figure}
}

\begin{figure}[t]
	\centering\includegraphics[trim=0.1in 0.1in 2.0in 0.0in, clip,width=0.57\textwidth]{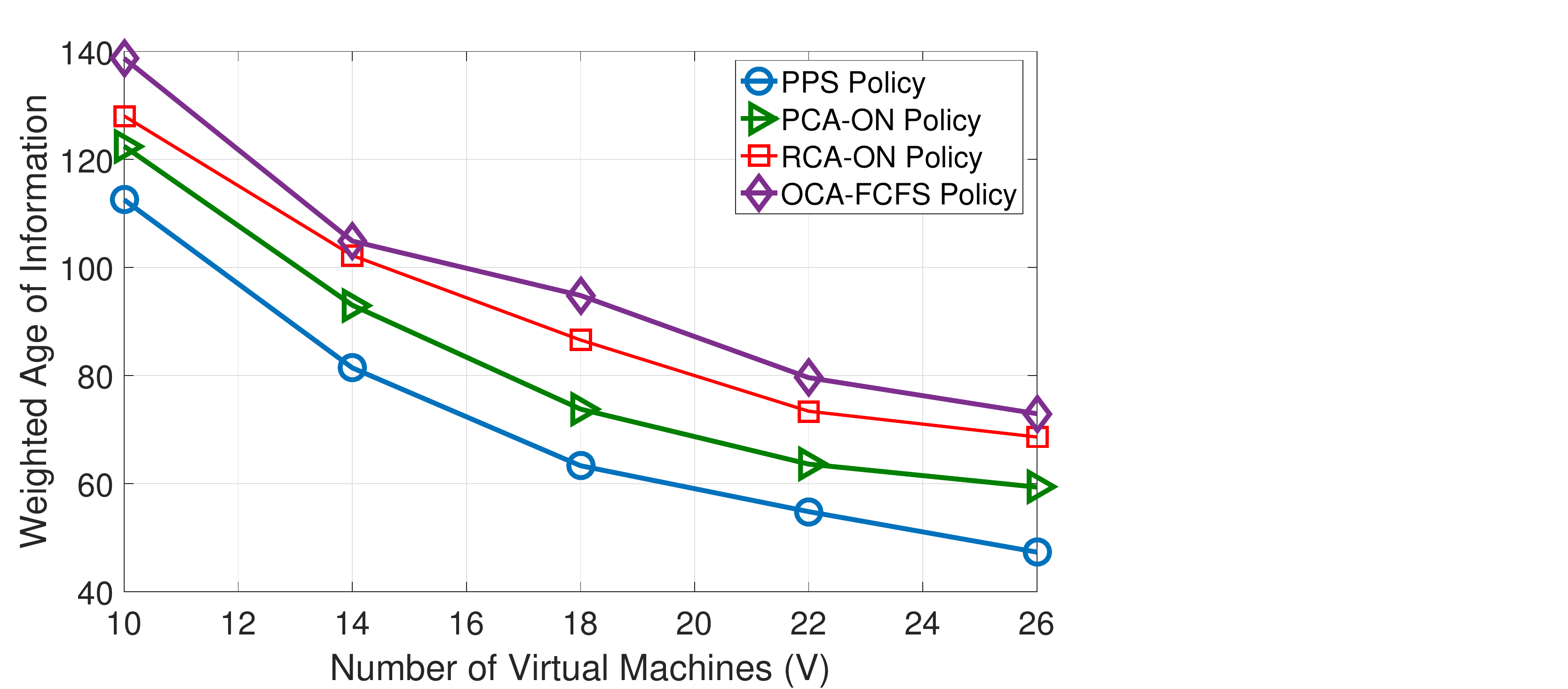}
	\caption{Weighted age of information for different number of virtual machines.  
		\label{AoI_vs_V}}
\end{figure}

\begin{figure}[t]
	\centering\includegraphics[trim=0.1in 0.1in 3.15in 0.0in, clip,width=0.55\textwidth]{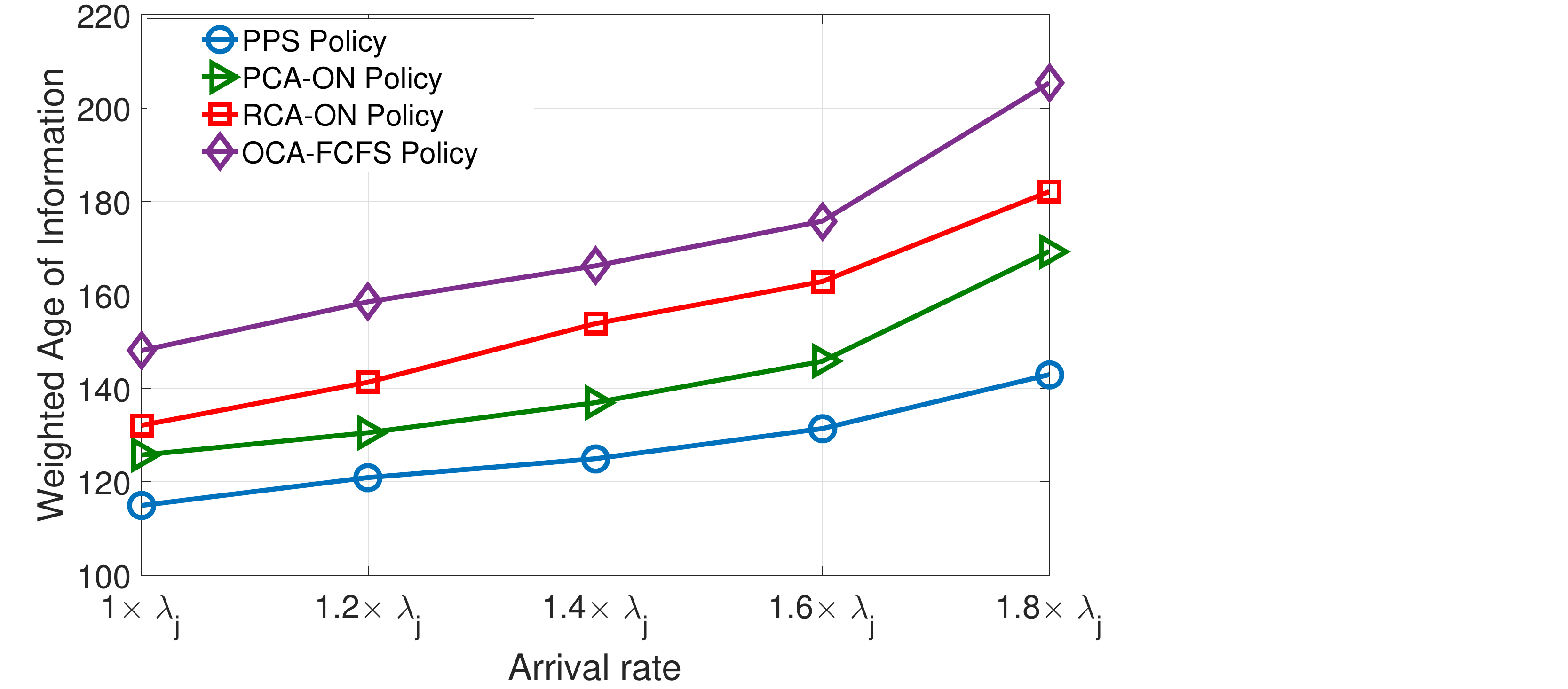}
	\caption{  Weighted age of information for different arrival rates $\lambda_j$. We vary the basic arrival rate from $1\times\lambda_j$ to $1.8 \lambda_j$ with an increment step of $0.20$.  
		\label{AoI_vs_Lambda}}
\end{figure} 

\begin{figure}[t]
	\centering\includegraphics[trim=0.1in 0.1in 3.15in 0.0in, clip,width=0.52\textwidth]{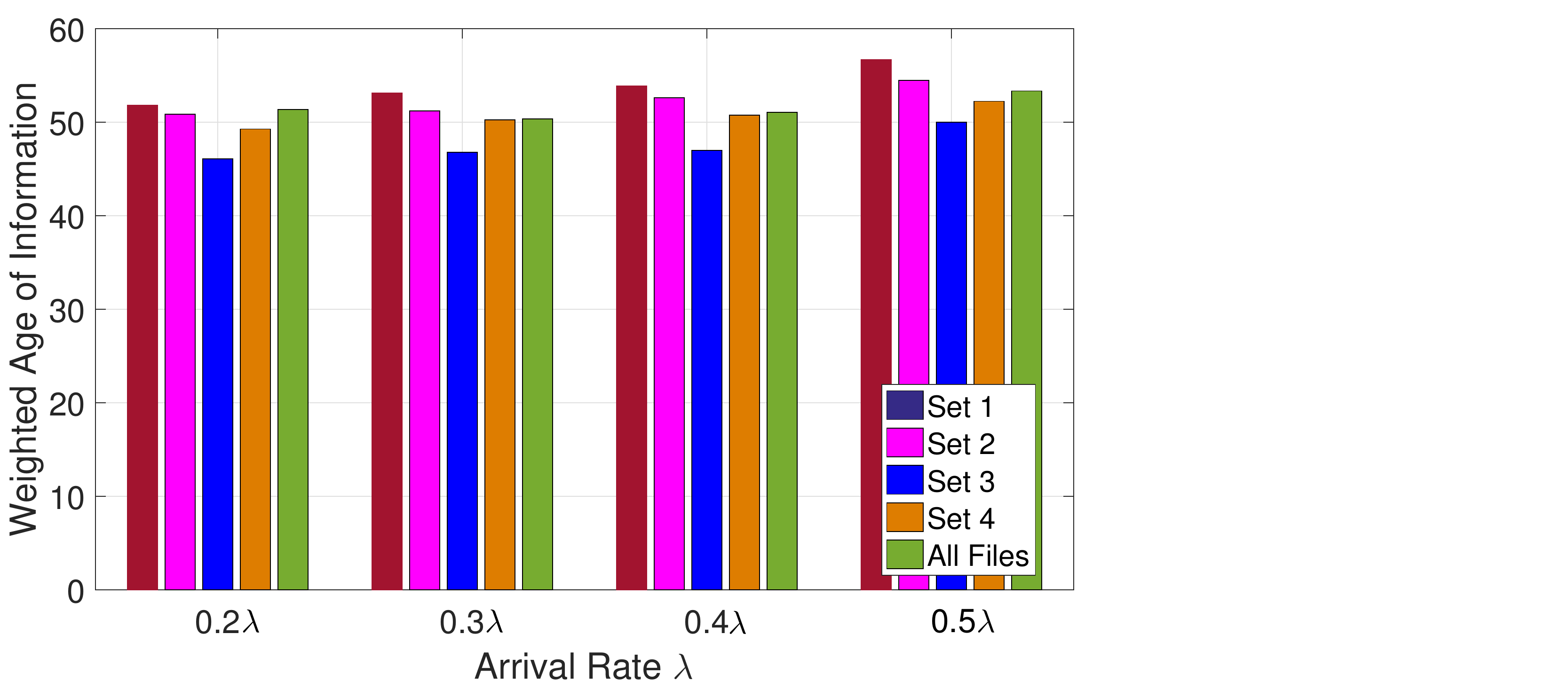}
	\caption{ \textcolor{black}{ Weighted age of information for different set of arrival rates $\lambda_j$. We vary the arrival rate of files from $0.2\times\lambda_j$ to $.5 \lambda_j$ with an increment step of $0.10\lambda_j$, where $\lambda_j$ is the base arrival rate.  We set the number of VMs to $14$.
		\label{AoI_groups}}}
\end{figure}

\begin{figure}[t]
	\centering\includegraphics[trim=0.0in 0.0in 4.4in 0.0in, clip,width=0.50\textwidth]{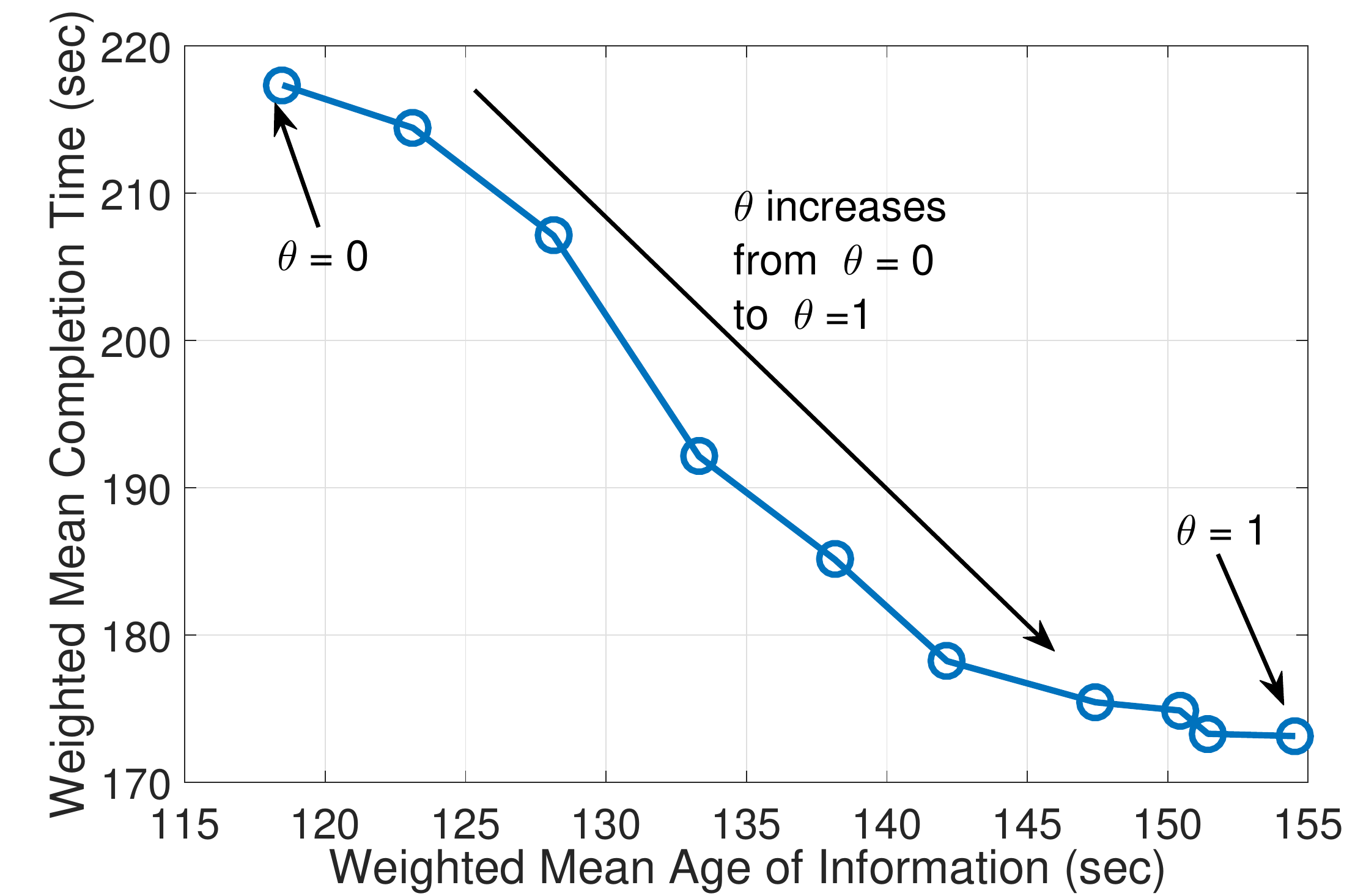}
	\caption{  Tradeoff between mean completion time and average age of information obtained by varying $\theta$.  
		\label{tradeoff}}
\end{figure} 

\textcolor{black}{
{\it Convergence of the Proposed Algorithm:}
Figure \ref{conveg} shows that
the proposed algorithm converges within
200 iterations to the optimal value, validating
the efficiency of the proposed optimization algorithm. It shows the convergence of the weighted sum of AoI and completion time versus the number of iterations. In the rest of the results, $200$
iterations will be used to get the required results.}

{\it Impact of number of VMs:}
The impact of changing the number of VMs is captured in Figure \ref{AoI_vs_V}. The number of VMs is changed from 8 VM to 24 VM and we set $\theta = 1$ to focus on age minimization. Clearly, the proposed PPS Policy, where assigning jobs to VMs is optimized, performs the best followed by RCA-ON and OCA-FCFS policy performs the worst. 
It can be seen that as the number of VMs increases, the AoI metric is improved (the received information is more recent with small age). Further, as the number of VMs decreases, the percentage of improvement increases. This observation is because with less VMs, more priority classes are existed and also more requests are expected to wait in the second queues. Thus, the relative impact of efficient machine scheduling can accordingly increase. Thus, PPS policy should have larger improvement when the number of VMs is limited, which is the typical case in real systems. 

{\it Effect of arrival rate of update requests:} Figure \ref{AoI_vs_Lambda} shows the effect of increasing system workload, obtained by varying the arrival rates of the vehicle request for updates from $\lambda$ to $1.8\lambda$ with an increment step of $0.2$, where $\lambda$ is the base arrival rate, on the weighted age of information.
We see a significant improvement in the weighted AoI with the proposed strategy as compared to the baselines. For instance, at the
arrival rate of $1.8\lambda_j$, where $\lambda_j$ is the base arrival rate defined above, the proposed PPS strategy reduces the weighted AoI by more than $20\%$ as compared to the PCA-ON strategy. Further, while the weighted AoI increases as arrival rate increases, the PPS policy still able to maintain small age by better optimizing the system parameters. 

\textcolor{black}{
{\it Effect of the age of information weights:} We next show the effect of varying the weights (i.e., $w_i$,s) on the weighted age of information for the proposed PPS policy. We divide the arrival rate of files into four groups, each with different scale for the base arrival
rates $\lambda_i$. We set $\lambda_i=2/150$ for the first set, $\lambda_i=4/150$ for the second set of $200$ vehicles/requests, $6/250$ for the third set and lastly we set $\lambda_i=3/150$ for the last $200$ requests.
We vary the arrival rate of all files from $0.2\lambda_i$ to $0.5\lambda_i$ with
a step of $0.5\lambda_i$ and plot the weighted age for each group of $200$ files as well as the overall value in Figure \ref{AoI_groups}. While weighted AoI increases as arrival rate increases, our PPS algorithm assigns differentiated
age for different file groups. Group $3$ that has highest weight $3$ (i.e., most age sensitive) always receive
the minimum age tail even though these
files have the highest arrival rate. Thus, efficiently reducing the age of the high arrival rate files reduces the overall
weighted age of information. We note that efficient access probabilities help in differentiating information age
as compared to the strategy where minimum queue-length servers are selected to access the content obtaining lower
weighted age.
}

{\it Tradeoff between Completion time and AoI:}
The preceding analysis and results show a trade off between the mean age of information and
the  average completion time of the jobs. In order to investigate such tradeoff, Figure \ref{tradeoff} plots the  average age of information versus the mean completion time for different values of $\theta$ ranging from $\theta=0$ to $\theta=1$. This figure implies that a compromise between the two  metrics can be achieved by our proposed PPS scheduling algorithm by setting $\theta$  to an appropriate value. As expected, increasing $\theta$ will increase the mean age of information as there is more priority to minimize the average completion time.  More importantly, this plot serves as a look-up for the service provider to decide on an efficient trade-off point between the two metrics based on a desired performance level. Thus, an efficient tradeoff point between the two metrics can be chosen based on the  level desired by the application, i.e.,  a tolerable level of staleness in the delivered information.

\textcolor{black}{
	\begin{figure}[t]
		\centering\includegraphics[trim=0.2in 0.07in 4.0in 0.0in, clip,width=0.50\textwidth]{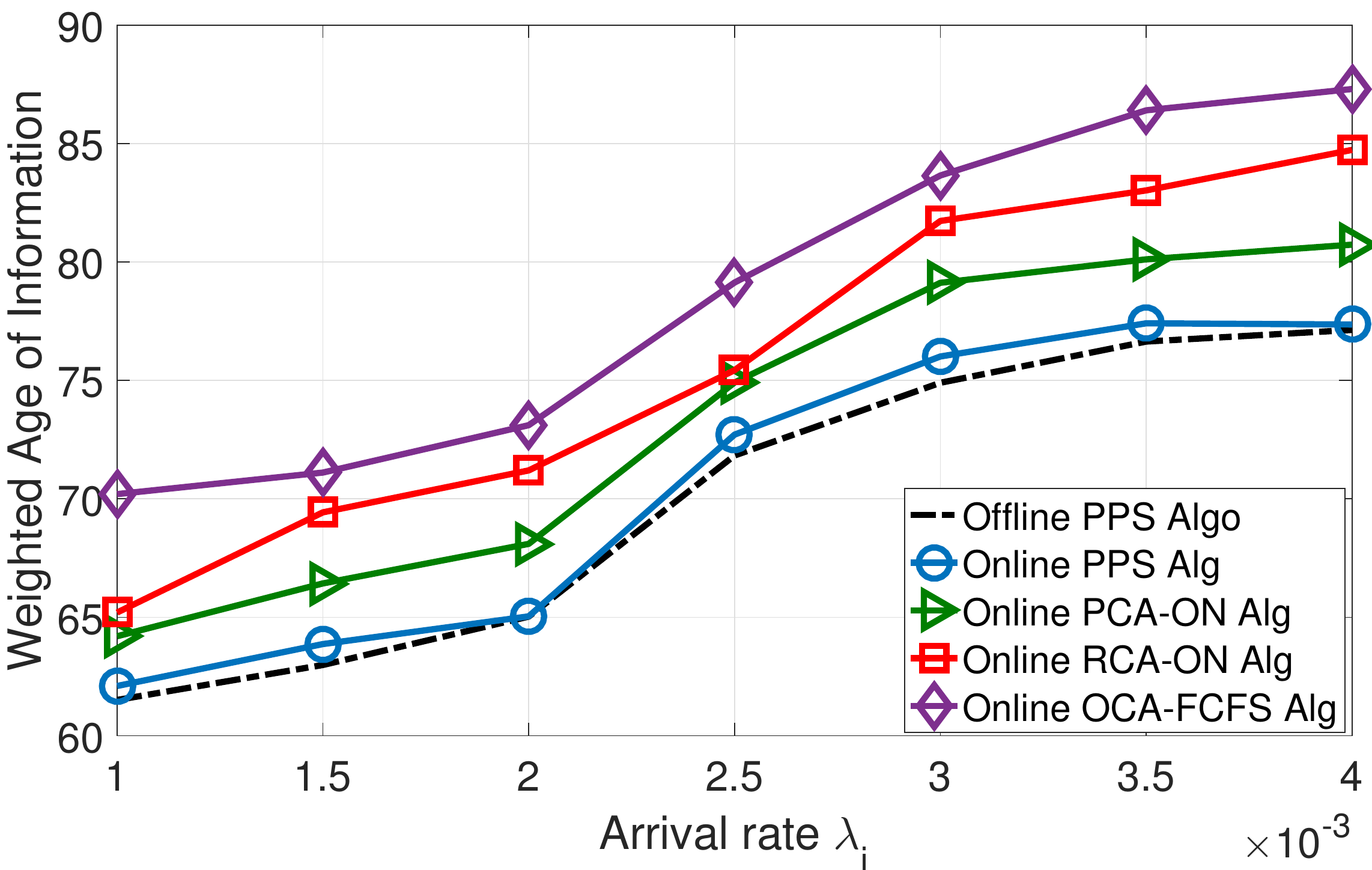}
		\caption{ \textcolor{black}{ Weighted age of information for different arrival rates $\lambda_j$, and $V=13$.  
			\label{AoI_online}}}
	\end{figure}
{\it Performance of the online algorithms:}
We study the performance of our online PPS algorithm and compare it with different baselines in Figure \ref{AoI_online}. We use MSR Cambridge public Traces  \cite{narayanan2008write}, which are week long block I/O traces of enterprise servers at Microsoft,  to test our proposed PPS algorithm and compare it with the different baselines. We test our algorithms on a data which combines the requests on volume 1 of media server and the requests on volume 1 of web server. The number of read requests from these volumes are 143,973 and 606,487, respectively. We take the file size as a unique identifier, since file id is not given in the data set. We first note that when the system workload increases, the AoI increases. However, PPS still achieves the lowest age by efficiently exploiting all the design control parameters including the scheduling probabilities and priority-based preference. Further, we also note that our online algorithm does not diverge from (close to) the offline PPS algorithm and thus validates the superiority of our proposed PPS algorithm. Moreover, while our PPS algorithm optimizes the system parameters offline, this figure shows that an online version of our algorithm can be developed to keep track of the systems dynamics and thus achieve an improved performance. 
}



\section{Conclusions}
\label{conc}

This paper aims to optimize the freshness and completion time for jobs requested by IoT nodes where different IoT devices are updating information periodically. Novel scheduling strategies are proposed for optimizing the system performance. Based on these strategies, mean age of information and average completion time are calculated. Then, we used those expression to formulate an optimization problem that minimizes a convex function of the two objectives. The problem is shown to be convex and thus an optimal solution is provided.  Based on the offline PPS version, an online algorithm is developed.  Numerical results demonstrate significant improvement as compared to the considered baselines. Further, the results provide important design guidelines for service providers in an IoT networks to provide a  desirable level of staleness in the delivered information. 


\bibliographystyle{IEEEtran}

\bibliography{ref_AoI}

\newpage
\appendices
\setcounter{page}{1}


%


\end{document}